\documentclass[12pt]{article}
\usepackage{graphicx,amsmath,amssymb}
\usepackage{fullpage}
\usepackage{xcolor}
\usepackage{subfig}
\usepackage{amsmath,amsfonts,amsthm}
\newtheorem{theorem}{Theorem} 
\newtheorem{lemma}[theorem]{Lemma} 
\newtheorem{cor}[theorem]{Corollary} 

\newcommand{\change}[1]{{#1}}
\newcommand{\note}[1]{{#1}}

\newcommand{\Fnote}[1]{{#1}} 
\newcommand{\FNote}[1]{{#1}} 

\begin{document}


\title{Morphing Planar Graph Drawings with Unidirectional Moves}
\author{Fidel Barrera-Cruz\thanks{University of Waterloo, Waterloo,
Canada, {\tt \{fbarrera, pehaxell,
  alubiw\}@uwaterloo.ca}}\ \thanks{Partially supported by CONACyT.} \and Penny
Haxell\footnotemark[1]\ \thanks{Partially supported by NSERC.} \and
Anna Lubiw\footnotemark[1]\ \footnotemark[3] }


\maketitle

\begin{abstract}

  Alamdari et al.~\cite{poly-morph} showed that given two
  straight-line planar drawings of a graph, there is a morph between
  them that preserves planarity and consists of a polynomial number of
  steps where each step is a \emph{linear morph} that moves each
  vertex at constant speed along a straight line.  \Fnote{An important
    step in their proof consists of converting a \textit{pseudo-morph}
    (in which contractions are allowed) to a true morph. Here we
    introduce the notion of
  \emph{unidirectional morphing} step,} where the vertices move along
lines that all have the same direction.
\FNote{Our main result is to show that any planarity preserving
  pseudo-morph consisting of unidirectional steps and contraction of
  low degree vertices can be turned into a true morph without
  increasing the number of steps}. Using this, we strengthen Alamdari
  et al.'s result to use only unidirectional morphs, and in the
  process we simplify the proof.
\end{abstract}

\section{Introduction}

\Fnote{Intuitively a morph can be thought of as a continuous
  deformation between structures. Morphs have been of interest in
  computer graphics~\cite{Gomes99}. Morphing can be used in the areas
  of medical imaging and geographical information
  systems~\cite{Barequet96,Barequet04} for reconstructing a surface
  given a sequence of parallel slices. Here we are interested in
  morphs restricted to graph drawings.}  

A \emph{morph} between two planar drawings $\Gamma_1$ and $\Gamma_2$
of a graph $G$ is a continuous family of drawings indexed by time
$t \in [0,1]$ where the drawing at time $t=0$ is $\Gamma_1$ and the
drawing at time $t=1$ is $\Gamma_2$.  A morph \emph{preserves
  planarity} if all intermediate drawings are planar.  Of necessity,
this means that $\Gamma_1$ and $\Gamma_2$ are topologically
equivalent, i.e., have the same faces and the same outer face.  A
morph \emph{preserves straight-line planarity} if all intermediate
drawings are straight-line planar, in which case the morph is a
continuous movement of the vertices of the graph drawing.

In 1944 Cairns~\cite{Cairns} proved the existence of a straight-line
planarity preserving morph between any two topologically equivalent
straight-line planar drawings of a triangulated graph.  Cairns' proof
is constructive but the resulting morph takes an exponential number of
steps.  Thomassen~\cite{T} extended the result to general
straight-line planar drawings by augmenting both drawings to
isomorphic triangulations, later called ``compatible''
triangulations~\cite{ASS}.  For a graph of $n$ vertices, compatible
triangulations can be found in polynomial time and have size $O(n^2)$
and this bound is tight. \change{Floater and
  Gotsman~\cite{Floater1999} gave a polynomial time algorithm
using Tutte's graph drawing algorithm~\cite{Tutte1963}, but in their
morph the trajectories of the vertices are complicated.  

\Fnote{Morphs preserving other aspects have also been studied. For
  example, the method in~\cite{Floater1999} was generalized by Gotsman
  and Surazhsky~\cite{Gotsman2001} to obtain a morph between two
  simple polygons that preserves simplicity. The existence of
  intersection-free morphs of maximal planar graphs has been
  established for certain types of sphere drawings by Kobourov and
  Landis~\cite{Kobourov2008}. Biedl et al. present an algorithm to
  morph between any two planar and orthogonal graph drawings while
  preserving planarity and orthogonallity using a polynomial number of
  linear morphing steps in~\cite{biedl2013}.

  Lubiw and Petrick~\cite{LP} showed that given two planar drawings of
  a graph there exists a planarity preserving morph that consists of
  polynomially many linear morphs where edges are allowed to
  bend. After this result was proven it was still unknown whether it
  was possible to morph between drawings of maximal planar graphs
  while preserving planarity in a polynomial number of steps.}
Recently, Alamdari et al.~\cite{poly-morph} gave an algorithm, based
on Cairns' approach, that solves the problem using $O(n^{2})$
\emph{linear morphs}, a morph that moves each vertex along a straight
line at uniform speed.  Using compatible triangulations this gives a
morph of $O(n^4)$ steps for general planar graphs.} 
 
In this paper we improve the result of Alamdari et al.~on morphing
triangulations in two ways: (1) we give a simpler proof; and (2) our
elementary steps are \emph{unidirectional morphs}.  A
\emph{unidirectional morph} is a linear morph where every vertex moves
parallel to the same line, i.e.~there is a line $L$ with unit
direction vector $\bar \ell$ such that each vertex $v$ moves at
constant speed from initial position $v_0$ to position
$v_0 + k_v \bar \ell$ for some $k_v \in {\mathbb R}$.  Note that $k_v$
may be positive or negative and that different vertices may move
different amounts along direction $\bar \ell$.  We call this an
\emph{$L$-directional morph}.
\Fnote{Our main contribution is to show that any planar preserving
  pseudo-morph consisting of unidirectional steps and contraction of
  low degree vertices can be turned into a true morph without
  increasing the number of steps. Using this and following the
  approach of Alamdari et al. we obtain a morph which is simpler and
  requires the same number of morphing steps, namely $O(n^{2})$.}
\FNote{ Very recently, our result was used by Angelini et al.~\cite{Angelini14} to give
  an improved morphing algorithm that uses only $O(n)$ unidirectional steps.}


In the remainder of this section we describe the high-level idea of
our result.

The existence proof of Cairns works by successively \emph{contracting}
a vertex of degree at most 5 to a neighbour, i.e.~moving the vertex
along one of its incident edges until it reaches the other endpoint of
the edge.  (The relevance of low degree is that a vertex of degree at
most 5 always has a neighbour to which it can be contracted while
preserving planarity.)  Each such step is a unidirectional morph, for
the trivial reason that only one vertex moves.  The number of steps is
exponential.  The result is not a true morph since vertices become
coincident, but Cairns argues that each vertex can be moved close to,
but not coincident with, the target vertex.  This fix causes a further
exponential increase in the number of steps.

Alamdari et al.~improved the number of morphing steps to a polynomial
number using the same two-phase approach.  The first phase finds a
\emph{pseudo-morph} which is defined as a sequence of the following
kinds of steps:
\begin{itemize}
\item{} a linear morph
\item{} a contraction of a vertex $p$ to another vertex, followed by a
pseudo-morph between the two reduced drawings, and then an
``uncontraction'' of $p$.
\end{itemize}

The number of steps in a pseudo-morph is defined to be the number of
linear morphs plus the number of contractions and uncontractions.
Alamdari et al.~give a pseudo-morph of $O(n^2)$ steps.

In the second phase they convert the pseudo-morph to a true morph that
avoids coincident vertices.  This requires a somewhat intricate
geometric argument that instead of contracting a vertex $p$ to a
neighbour, it is possible to move $p$ close to the neighbour and keep
it close during subsequent morphing steps without increasing the
number of steps.

\Fnote{Here we use a different approach for the second phase, which results
in a simpler proof and uses only unidirectional morphs. This is in
Section~\ref{sec:coincident}. To obtain our strengthened version of Alamdari et
al.'s result, we use essentially the same pseudo-morph
for the first stage. We must verify that unidirectional morphs suffice. This
is described in Section~\ref{sec:pseudo}.}

We use the following notation.  If $\Gamma_1, \ldots, \Gamma_k$ are
straight-line planar drawings of a graph, then $\langle \Gamma_1,
\ldots, \Gamma_k \rangle$ denotes the morph that consists of the $k-1$
linear morphs from $\Gamma_i$ to $\Gamma_{i+1}$ for $i=1, \ldots,
k-1$.

\section{A pseudo-morph with unidirectional morphing steps}
\label{sec:pseudo}

Alamdari et al.~\cite{poly-morph} give a pseudo-morph of $O(n^2)$
steps to go between any two topologically equivalent straight-line
planar drawings of a triangulated graph on $n$ vertices.  In this
section we show that their pseudo-morph can be implemented with
unidirectional morphs.  They show that the only thing that is needed
is a solution to the following problem using $O(n)$ linear morphs:

\medskip\noindent PROBLEM 3.2.~(4-GON CONVEXIFICATION) Given a
triangulated graph G with a triangle boundary and a 4-gon $abcd$ in a
straight-line planar drawing of $G$ such that neither $ac$ nor $bd$ is
an edge outside of $abcd$ (i.e.,~$abcd$ does not have external
chords), find a pseudo-morph so that $abcd$ becomes convex.  \medskip

\begin{figure}[h]
  \label{fig:4-gon}
  \centering
  \includegraphics[width=1in]{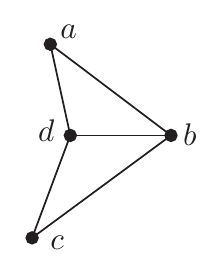}
  \caption{A 4-gon $abcd$.}
\end{figure}

The main idea is to use Cairns' approach: find a low-degree vertex,
contract it to a neighbour, and recurse on the resulting smaller
graph.  Each such contraction is a unidirectional morph.  This
approach works so long as there is a low-degree vertex that is not a
\emph{problematic vertex} $p$ defined as follows:

\begin{enumerate}
\item{} $p$ is a vertex of the boundary triangle $z_1, z_2, z_3$.
\item{} $p$ is a vertex of the 4-gon $abcd$ and is not on the
  boundary.
\item{} $p$ is outside the 4-gon, is not on the boundary, has degree
  at most 5, and is adjacent to both $a$ and $c$, and either $a$ or
  $c$ is in the kernel of the polygon formed by the neighbours of $p$.
  (In this case contracting $p$ to $a$ or $c$ would
  create the edge $ac$ outside the 4-gon.)
\end{enumerate}

Alamdari et al.~show how to handle each type of problematic vertex.
We must go through the cases and argue that unidirectional morphs
suffice in each case.

Problematic vertices of type (2) and type (3) are handled (in their
Sections 4.1 and 4.3) by moving a single vertex at a time either by
contracting or by moving a vertex very close to another vertex.
Moving a single vertex is a unidirectional morph so these cases are
done.

It remains to consider problematic vertices of type (1) which they do
in Section 4.2.  To handle this case they use an operation where one
vertex of a triangle moves along a straight line and the other
vertices inside the triangle follow along linearly.  We will show that
the motion is in fact unidirectional.  Because we will need it later
on, we will consider a more general situation where all three vertices
of the triangle undergo a unidirectional morph.

\begin{lemma} Let $a,b,c$ be the vertices of a triangle and let $x$ be
a point inside the triangle defined by the convex combination $
\lambda_1 a + \lambda_2 b + \lambda_3 c$ where $\sum \lambda_i = 1$
and $\lambda_i \ge 0$.  If $a$, $b$, and $c$ move linearly in the
direction of the vector $\bar{d}$ then so does $x$.
  \label{lemma:convex-comb}
\end{lemma}

\begin{proof}
  Suppose the morph is indexed by $t \in [0,1]$ and that the positions
  of the vertices at time $t$ are $a_t, b_t, c_t, x_t$.  Suppose that
  $a$ moves by $k_1 \bar d$, and $b$ moves by $k_2 \bar d$, and $c$
  moves by $k_3 \bar d$.  Thus $a_t = a_0 + tk_1 \bar{d}$ and etc.
  Then $$x_t = \lambda_1 a_t + \lambda_2 b_t + \lambda_3 c_t \\ =
  \lambda_1 a_0 + \lambda_2 b_0 + \lambda_3 c_0 + t(\lambda_1 k_1 +
  \lambda_2 k_2 + \lambda_3 k_3) \bar{d} \\= x_0 + tk \bar{d}$$ where
  $k = \lambda_1 k_1 + \lambda_2 k_2 + \lambda_3 k_3$.  Thus $x$ also
  moves linearly in direction $\bar d$.
\end{proof}

Alamdari et al.~show that to handle problematic vertices of type (1)
it suffices to handle the following two cases.  We repeat their
arguments, adding details of exactly how Lemma~\ref{lemma:convex-comb}
gives unidirectional morphs.
\begin{enumerate}
\item[A.] There is a boundary vertex, say $z_1$, of degree 3.  See
Figure~\ref{fig:4-gon-deg-3}(a).

Then $z_1$ must have a neighbour $y$ that is adjacent to $z_2$ and
$z_3$.  If $abcd$ lies entirely inside the triangle $T=yz_2z_3$ then
we recursively morph the subgraph contained in $T$.  Otherwise $abcd$
must include at least one triangle outside $T$.  It cannot have both
triangles outside $T$ because of the assumption that there is no edge
$ac$.  Thus we can assume without loss of generality that $abcd$
consists of triangle $z_1yz_2$ and an adjacent triangle inside $T$,
see Figure~\ref{fig:4-gon-deg-3}(b).  The solution is to move $y$
towards $z_1$ to directly convexify $abcd$.  As $y$ is moved, the
contents of triangle $T$ follow along linearly.  By
Lemma~\ref{lemma:convex-comb} this is a unidirectional morph.
 
\begin{figure}[h]
  \centering
  \includegraphics[width=4.5in]{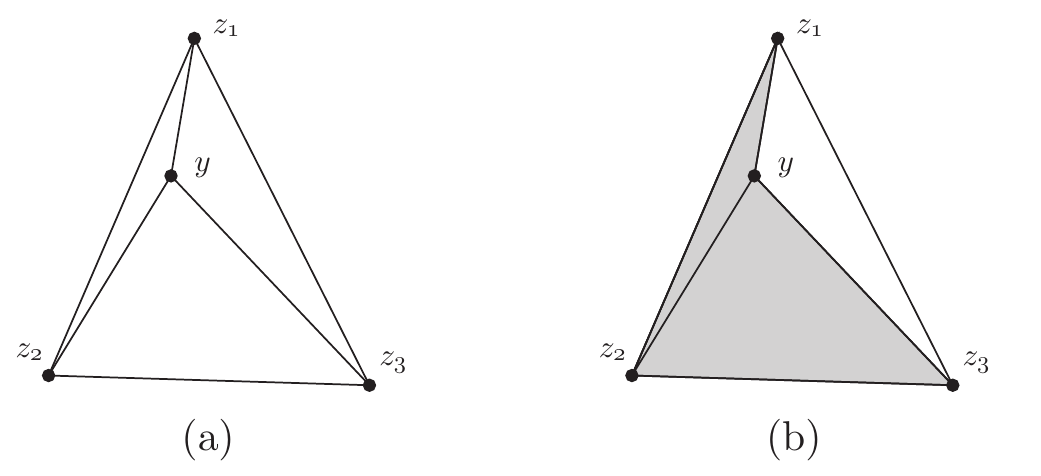}
  \caption{A boundary vertex of degree 3.}    \label{fig:4-gon-deg-3}
\end{figure}

\item[B.] All three boundary vertices have degree 4.  See
Figure~\ref{fig:4-gon-deg-4}(a).

In this case there must exist an internal triangle $T= y_1 y_2 y_3$
containing all the internal vertices with $y_i$ adjacent to $z_j$ for
$j \ne i, i,j \in \{1,2,3\}$.
If $abcd$ lies entirely inside $T$ then we recursively morph the
subgraph contained in $T$.  If $abcd$ lies entirely outside $T$ then
without loss of generality $abcd$ is $z_2z_1y_2y_3$.  The solution is
to move $y_2$ towards $z_1$ to directly convexify the
4-gon.  
As $y_2$ is moved, the contents of triangle $T$ follow along linearly.
By Lemma~\ref{lemma:convex-comb} this is a unidirectional morph.

  The final possibility is that $abcd$ consists of one triangle
outside $T$ and one triangle inside $T$.  We may assume without loss
of generality that $abcd$ consists of triangle $z_1y_3y_2$ and an
adjacent triangle inside $T$, see Figure~\ref{fig:4-gon-deg-4}(b).  We
will convexify the 4-gon $z_1 y_2 y_1 y_3$, which will necessarily
also convexify $abcd$.  The solution is to move $y_1$ towards $z_2$ to
convexify the 4-gon while the contents of triangle $T$ follow along
linearly.  By Lemma~\ref{lemma:convex-comb} this is a unidirectional
morph.

  \begin{figure}[h] \centering
    \includegraphics[width=4.5in]{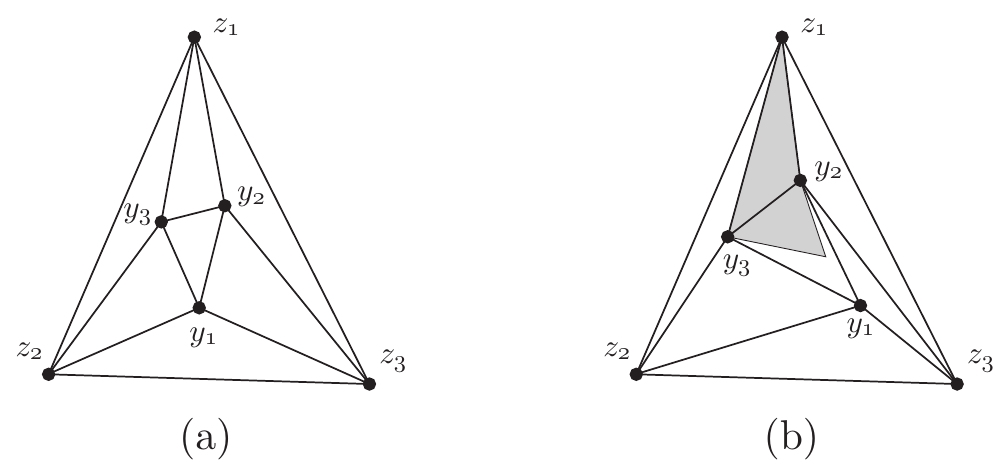}
    \caption{All three boundary vertices of degree 4.}
    \label{fig:4-gon-deg-4}
  \end{figure}

\end{enumerate}

This completes the argument that the pseudo-morph of Alamdari et
al.~can be implemented with unidirectional morphs.

\vspace{3mm}
\section{Avoiding coincident vertices}
\label{sec:coincident}

\Fnote{In this section we describe our key lemma for converting a
  pseudo-morph to a true morph that avoids coincident vertices.

\begin{lemma}\label{lem:true_morph}
  Let $\mathcal{M}$ be a pseudo-morph between drawings $\Gamma_{1}$
  and $\Gamma_{2}$ consisting of $k$ planar unidirectional steps that
  acts on a triangulated planar graph $G$. If only vertices of degree
  at most $5$ are contracted in $\mathcal{M}$, then there
  exists a planar morph consisting of $k$ unidirectional steps from
  $\Gamma_{1}$ to $\Gamma_{2}$.
\end{lemma}}

\Fnote{We now outline the proof of Lemma~\ref{lem:true_morph} by again
  following the approach of Alamdari et al.~\cite{poly-morph}}.
Suppose the pseudo-morph consists of the contraction of a non-boundary
vertex $p$ of degree at most 5 to a neighbour $a$, followed by a
pseudo-morph $\cal M$ of the reduced graph and then an uncontraction
of $p$.  The pseudo-morph $\cal M$ consists of unidirectional morphing
steps and by induction we can convert it to a morph $M$ that consists
of the same number of unidirectional morphing steps.  We will show how
to modify $M$ to $M^p$ by adding $p$ and its incident edges back into
each drawing of the morph sequence.  To obtain the final morph, we
replace the contraction of $p$ to $a$ by a unidirectional morph that
moves $p$ from its initial position to its position at the start of
$M^p$, then follow the steps of $M^p$, and then replace the
uncontraction of $p$ by a unidirectional morph that moves $p$ from its
position at the end of $M^p$ to its final position.  The result is a
true morph that consists of unidirectional morphing steps and the
number of steps is the same as in the original pseudo-morph.

Thus our main task is to modify a morph $M$ to a morph $M^p$ by adding
a vertex $p$ of degree at most 5 and its incident edges back into each
drawing of the morph sequence, maintaining the property that each step
of the morph sequence is a unidirectional morph.  It suffices to look
at the polygon $P$ formed by the neighbours of $p$.  We know that $P$
has a vertex $a$ that remains in the kernel of $P$ throughout the
morph.  We will place $p$ near $a$.  We separate into the cases where
$P$ has 3 or 4 vertices, which are quite easy, and the case where $P$
has 5 vertices, which is more involved.  The following two lemmas
handle these two cases, and together strengthen Lemma 5.2
of~\cite{poly-morph} by adding the unidirectional condition.  

\begin{lemma} Let $P$ be a $\le 4$-gon and let $\Gamma_1, \ldots,
\Gamma_k$ be straight-line planar drawings of $P$ such that each morph
$\langle \Gamma_i, \Gamma_{i+1} \rangle, i=1, \ldots, k-1$ is
unidirectional and planar, and vertex $a$ of $P$ is in the kernel of
$P$ at all times during the whole morph $\langle \Gamma_1, \ldots,
\Gamma_k \rangle$.  Then we can augment each drawing $\Gamma_i$ to a
drawing $\Gamma_i^p$ by adding vertex $p$ at some point $p_i$ inside
the kernel of the polygon $P$ in $\Gamma_i$ and adding straight line
edges from $p$ to each vertex of $P$ in such a way that each morph
$\langle \Gamma_i^{p}, \Gamma_{i+1}^{p} \rangle$ is unidirectional and
planar.
  \label{lemma:morph-aug-4}
\end{lemma}

\begin{proof}
  If $P$ is a triangle then by Lemma~\ref{lemma:convex-comb} we can
  place $p$ at a fixed convex combination of the triangle vertices in
  all the drawings $\Gamma_i$.

  If $P$ is a 4-gon $abcd$ then the line segment $ac$ also stays in
  the kernel, so we can place $p$ at a fixed convex combination of $a$
  and $c$ in all the drawings $\Gamma_i$ (using the degenerate version
  of Lemma~\ref{lemma:convex-comb} where the triangle collapses to a
  line segment).
\end{proof}

\begin{lemma}\label{lem:five-gon}
  Let $P$ be a $5$-gon and let $\Gamma_1, \ldots, \Gamma_k$ be
  straight-line planar drawings of $P$ such that each morph $\langle
  \Gamma_i, \Gamma_{i+1} \rangle, i=1, \ldots, k-1$ is unidirectional
  and planar, and vertex $a$ of $P$ is in the kernel of $P$ at all
  times during the whole morph $\langle \Gamma_1, \ldots, \Gamma_k
  \rangle$.  Then we can augment each drawing $\Gamma_i$ to a drawing
  $\Gamma_i^p$ by adding vertex $p$ at some point $p_i$ inside the
  kernel of the polygon $P$ in $\Gamma_i$ and adding straight line
  edges from $p$ to each vertex of $P$ in such a way that each morph
  $\langle \Gamma_i^{p}, \Gamma_{i+1}^{p} \rangle$ is unidirectional and planar.
\end{lemma}

The proof of Lemma~\ref{lem:five-gon} is more involved.  Let $P$ be
the 5-gon $abcde$ labelled clockwise.  We assume that vertex $a$ is
fixed throughout the morph.  This is not a loss of generality because
if $a$ moves during an $L$-directional morph we can translate the
whole drawing back in direction $L$ so that $a$ returns to its
original position.  An $L$-directional morph composed with a
translation in direction $L$ is again an $L$-directional morph, and
planarity is preserved since the relative positions of vertices do not
change.

%
%
Observe that at any time instant $t$ during morph $\langle
\Gamma_1,\dots,\Gamma_k \rangle$ there exists an $\epsilon_{t} > 0$
such that the intersection between the disk $D$ centered at $a$ with
radius $\epsilon_{t}$ and the kernel of polygon $P$ consists of a
positive-area sector $S$ of $D$.  This is because $a$ is a vertex of
the kernel of $P$.  Let $\epsilon = \min_{t}\epsilon_{t}$ be the
minimum of $\epsilon_{t}$ among all time instants $t$ of the morph.

Fix $D$ to be the disk of radius $\epsilon$ centered at $a$.  In case
$a$ is a convex vertex of $P$, the sector $S$ is bounded by the edges
$ab$ and $ae$ and we call it a \emph{positive} sector.  See
Figure~\ref{fig:sectors}(a).
In case $a$ is a reflex vertex of $P$, the sector $S$ is bounded by
the extensions of edges $ab$ and $ae$ and we call it a \emph{negative}
sector.  See Figure~\ref{fig:sectors}(b).  More precisely, let $b'$
and $e'$ be points so that $a$ is the midpoint of the segments $bb'$
and $ee'$ respectively.  The negative sector is bounded by the
segments $ae'$ and $ab'$.  Note that when an $L$-directional morph is
applied to $P$, the points $b'$ and $e'$ also move at uniform speed in
direction $L$.

\begin{figure}[htb]
 \centering 
 \subfloat[]{\includegraphics[scale=1.7]{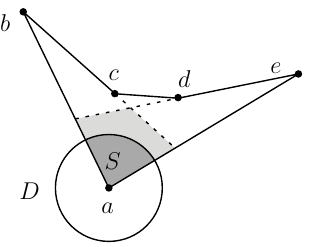}\label{fig:convex-disk}}
 \hspace{20pt}
 \subfloat[]{\includegraphics[scale=1.7]{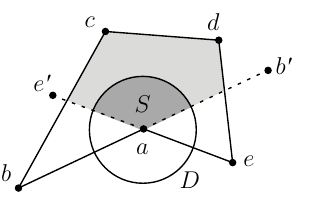}\label{fig:reflex-disk}}
 \caption{A disk $D$ centered at $a$ whose intersection with the
   kernel of $P$ (the lightly shaded polygonal region) is a non-zero-area
   sector $S$ (darkly shaded).  (a) Vertex $a$ is convex and $S$ is a
   positive sector.  (b) Vertex $a$ is reflex and $S$ is a negative
   sector.  }
 \label{fig:sectors}
\end{figure}

The important property we use from now on is that any point in the
sector $S$ lies in the kernel of polygon $P$.  Let the sector in
drawing $\Gamma_i$ be $S_i$ for $ i=1, \ldots, k$.  Let the direction
of the unidirectional morph $\langle \Gamma_i, \Gamma_{i+1} \rangle$
be $L_i$ for $ i=1, \ldots, k-1$.  In other words, $\langle \Gamma_i,
\Gamma_{i+1} \rangle$ is an $L_i$-directional morph.

Our task is to choose for each $i$ a position $p_i$ for vertex $p$
inside sector $S_i$ so that all the $L_i$-directional morphs keep $p$
inside the sector at all times.  A necessary condition is that the
line through $p_i p_{i+1}$ be parallel to $L_i$.  We will first show
that this condition is in fact sufficient (see
Lemma~\ref{lemma:p-move}).  Then we will show that such points $p_i$
exist.

Translate $L_i$ to go through point $a$ and distinguish the following
two cases in the relationship between $L_i$ and $S_i$:

\begin{description}
\item[one-sided case] Points $b_i$ and $e_i$ lie in the same closed
half-plane determined by $L_i$.  In this case, whether the sector
$S_i$ is positive or negative, $L_i$ does not intersect the interior
of $S_i$.  See Figure~\ref{fig:one-side}.  An $L_i$-directional morph
keeps $b_i$ and $e_i$ on the same side of $L_i$ so if $S_i$ is
positive it remains positive and if $S_i$ is negative it remains
negative.

\item[two-sided case] Points $b_i$ and $e_i$ lie on opposite sides of
$L_i$.  In this case $L_i$ intersects the interior of the sector
$S_i$.  See Figure~\ref{fig:two-side}.  During an $L_i$-directional
morph the sector $S_i$ may remain positive, or it may remain negative,
or it may switch between the two, although it can only switch once.
\end{description}

\begin{figure}[htb]
  \centering 
  \subfloat[]{\includegraphics[scale=1.5]{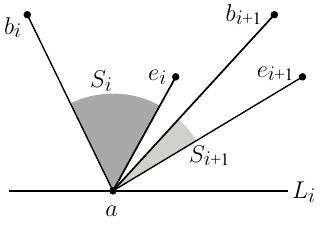}} \hspace{40pt} 
  \subfloat[]{\includegraphics[scale=2]{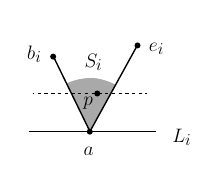}}
  \caption{The one-sided case where $S_i$ lies to one side of $L_i$,
    illustrated for a positive sector $S_i$.  (a) An $L_i$-directional
    morph to $S_{i+1}$.  (b) $p$ remains inside the sector iff it
    remains inside $D$ and between the two lines $ba$ and $ea$.  }
  \label{fig:one-side}
\end{figure}

\begin{figure}[htb]
  \centering 
  \subfloat[]{\includegraphics[scale=1.7]{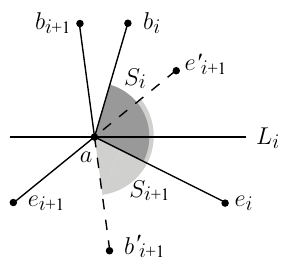}} \hspace{40pt}
  \subfloat[]{\includegraphics[scale=2]{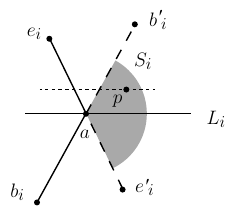}}
  \caption{The two-sided case where $S_i$ contains points on both
    sides of $L_i$.  (a) An $L_i$-directional morph from the positive
    sector $S_i$ bounded by $ b_i a e_i$ to the negative sector
    $S_{i+1}$ bounded by $ e'_{i+1} a b'_{i+1}$.  (b) $p$ remains
    inside the sector iff it remains inside $D$ and on the same side
    of the lines $bb'$ and $ee'$.  }
  \label{fig:two-side}
\end{figure}

Our main tool is the following lemma proving that ``sidedness'' on
line $L$ is preserved in an $L$-directional morph.

\begin{lemma}\label {lem:move_boundary} 
  Let $L$ be a horizontal line and $x_{0},x_{1},y_{0},y_{1}$ be points
  in $L$. Consider a point $x$ that moves at constant speed from
  $x_{0}$ to $x_{1}$ in one unit of time. If $y_{i}$ is to the right
  of $x_{i}$, $i=0,1$, and $y$ is a point that moves at constant speed
  from $y_{0}$ to $y_{1}$ in one unit of time then $y$ remains to the
  right of $x$ during their movement.  Note that $x_0$ may lie to the
  right or left of $x_1$ and ditto for $y_0$ and $y_1$.
\end{lemma}

\begin{proof}
  \begin{figure}[!h] \centering 
    \includegraphics{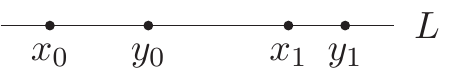}
    \caption{Points $x$ and $y$ move from $x_{0}$ to $x_{1}$ and from
      $y_{0}$ to $y_{1}$ respectively.}
    \label{fig:move_boundary}
  \end{figure}
  
  Let $x_{i}$ and $y_{i}$, $i=0,1$, be points as described above, see
  Figure~\ref{fig:move_boundary}.  Denote by $x_{t}$ and $y_{t}$ the
  positions of $x$ and $y$ for $0<t<1$. First note that
  \begin{equation}
    y_{i}=x_{i}+\delta_{i}\label{eq:yi}
  \end{equation}
  for $i=0,1$, with $\delta_{i}>0$. Since $x$ and $y$ are moving at
  constant speed, we have $x_{t}=(1-t)x_{0}+tx_{1}$ and
  $y_{t}=(1-t)y_{0}+ty_{1}$. Now, using equation~\eqref{eq:yi} in the
  expression for $y_{t}$ we have
  \begin{align}
    y_{t}&=(1-t)(x_{0}+\delta_{0})+t(x_{1}+\delta_{1})\\
&=x_{t}+(1-t)\delta_{0}+t\delta_{1},
  \end{align}
  where $(1-t)\delta_{0}+t\delta_{i}>0$.
\end{proof}

\begin{cor}
  \label{cor:line-side} Consider an $L$-directional morph acting on
points $p$, $r$ and $s$.  If $p$ is to the right of the line through
$rs$ at the beginning and the end of the $L$-directional morph, then
$p$ is to the right of the line through $rs$ throughout the
$L$-directional morph.
\end{cor}

We are now ready to prove our main lemma about the relative positions
of points $p_i$ and $p_{i+1}$.

\begin{lemma}
  \label{lemma:p-move} If point $p_i$ lies in sector $S_i$ and point
$p_{i+1}$ lies in sector $S_{i+1}$ and the line $p_i p_{i+1}$ is
parallel to $L_i$ then an $L_i$-directional morph from $S_i, p_i$ to
$S_{i+1}, p_{i+1}$ keeps the point in the sector at all times.
\end{lemma}

\begin{proof} 
  We use the notation that $b$ moves from $b_i$ to $b_{i+1}$, $p$
  moves from $p_i$ to $p_{i+1}$, etc.

  First consider the one-sided case.  Suppose $S_i$ is a positive
  sector (the case of a negative sector is similar). Observe that a
  point $p$ remains in the sector during an $L_i$-directional morph if
  and only if it remains in the disc $D$ and remains between the lines
  $ba$ and $ea$.  See Figure~\ref{fig:one-side}(b).  Because $p_i$ and
  $p_{i+1}$ both lie in disc $D$, thus the line segment between them
  lies in the disc, and $p$ remains in the disc throughout the morph.
  In the initial configuration, $p_i$ lies between the lines $b_i a$
  and $e_i a$, and in the final configuration $p_{i+1}$ lies between
  the lines $b_{i+1} a$ and $e_{i+1} a$.  Therefore by
  Corollary~\ref{cor:line-side} $p$ remains between the lines
  throughout the $L_i$-directional morph.  Thus $p$ remains inside the
  sector throughout the morph.

  Now consider the two-sided case.  Observe that a point $p$ remains
  in the sector during an $L_i$-directional morph if and only if it
  remains on the same side of the lines $b b'$ and $e e'$.  Note that
  this is true even when the sector changes between positive and
  negative.  See Figure~\ref{fig:two-side}(b).  As in the one-sided
  case, $p$ remains in the disc throughout the morph.  Also, $p$ is on
  the same side of the lines $b b'$ and $e e'$ in the initial and
  final situations, and therefore by Corollary~\ref{cor:line-side} $p$
  remains on the same side of the lines throughout the morph.  Thus
  $p$ remains inside the sector throughout the morph.
\end{proof}

With Lemma~\ref{lemma:p-move} in hand the only remaining issue is the
existence of points $p_i$.  We call the possible positions for $p_i$
inside sector $S_i$ the \emph{nice} points, defined formally as
follows:
\begin{itemize}
\item All points in the interior of $S_k$ are nice.
\item For $1 \leq i\leq k-1$, a point $p_i$ in the interior of $S_i$
is nice if there is a nice point $p_{i+1}$ in $S_{i+1}$ such that $p_i
p_{i+1}$ is parallel to $L_i$.
\end{itemize}

By Lemma~\ref{lemma:p-move} it suffices to show that all the nice sets
are non-empty.  We will in fact characterize the sets.  Given a line
$L$, an $L$-truncation of a sector $S$ is the intersection of $S$ with
an open slab that is bounded by two lines parallel to $L$ and
contains all points of $S$ in a small neighbourhood of $a$.  In
particular, an $L$-truncation of a sector is non-empty.

\begin{lemma}
  \label{lemma:nice} The set of nice points in $S_i$ is an
$L_i$-truncation of $S_i$ for $i=1, \ldots, k$.
\end{lemma}
\begin{proof} 
  Let $N_i$ denote the nice points in $S_i$.  The proof is by
  induction as $i$ goes from $k$ to 1.  All the points in the interior
  of $S_k$ are nice.  Suppose by induction that $N_{i+1}$ is an
  $L_{i+1}$-truncation of $S_{i+1}$.

  Consider the one-sided case.  See
  Figure~\ref{fig:transition_case_1_2}.  The slab determining $N_i$
  consists of all lines parallel to $L_i$ that go through a point of
  $N_{i+1}$.  $L_i$ itself forms one boundary of the slab and the slab
  is non-empty since $N_{i+1}$ contains all of $S_{i+1}$ in a small
  neighbourhood of $a$.  Thus the slab contains all points of $S_i$ in
  a neighbourhood of $a$, and thus $N_i$ is an $L_{i}$-truncation of
  $S_i$.

  \begin{figure}[!h] 
    \centering
    \includegraphics{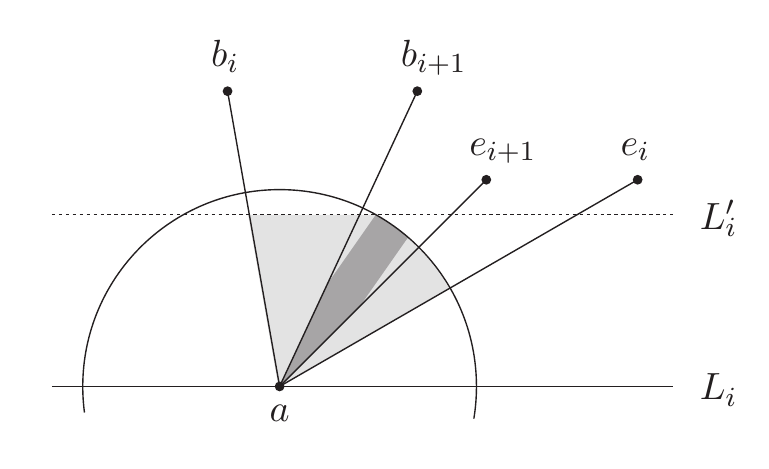}
    \caption{$N_i$ (lightly shaded) is an $L_i$-truncation of $S_i$ in
      the one-sided case.  $N_{i+1}$ is darkly shaded.  $L_i$ and
      $L'_i$ are the slab boundaries for $N_i$.  }
    \label{fig:transition_case_1_2}
  \end{figure}
  
  Consider the two-sided case.  See
  Figure~\ref{fig:pie_transition_case_2_1}.  The slab determining
  $N_i$ consists of all lines parallel to $L_i$ that go through a
  point of $N_{i+1}$.  The slab contains $a$ in its interior and thus
  $N_i$ is an $L_{i}$-truncation of $S_i$.

  \begin{figure}[htb] 
    \centering 
    \subfloat[]{\includegraphics[scale=.9]{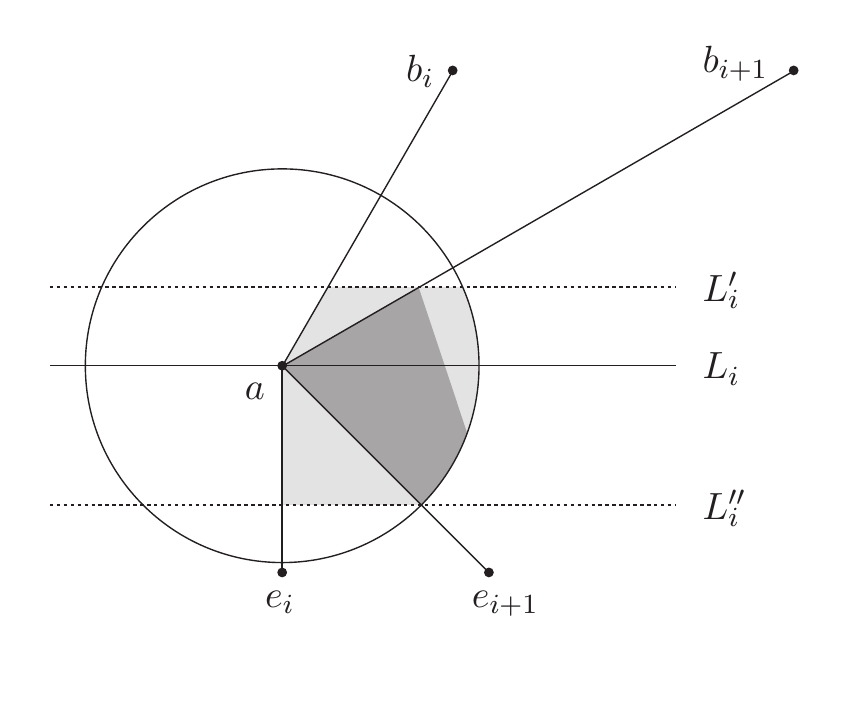}}
    \hspace{40pt} 
    \subfloat[]{\includegraphics[scale=.9]{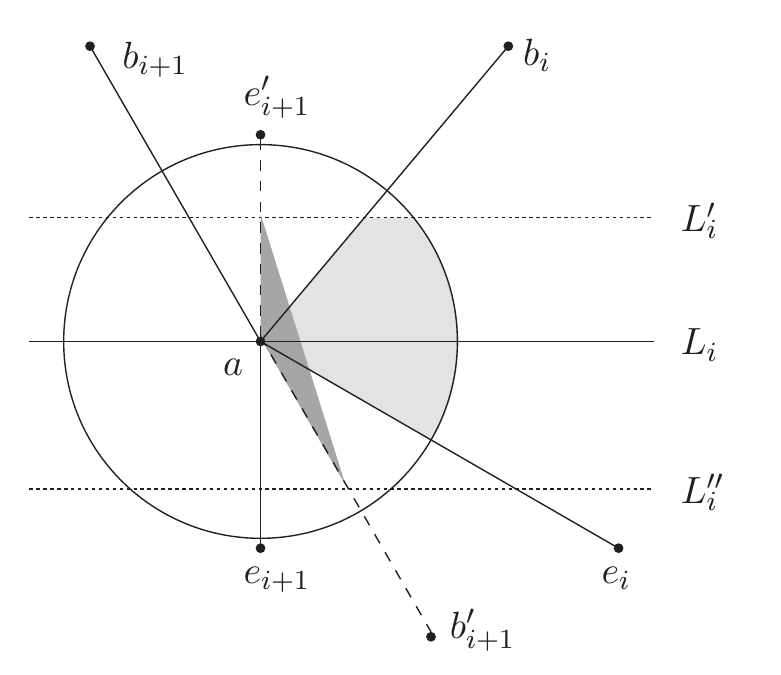}}
    \caption{$N_i$ (lightly shaded) is an $L_i$-truncation of $S_i$ in
      the two-sided case.  $N_{i+1}$ is darkly shaded.  $L'_i$ and
      $L''_i$ are the slab boundaries for $N_i$.  }
    \label{fig:pie_transition_case_2_1}
  \end{figure}

\end{proof}

Lemma~\ref{lemma:nice} implies in particular that the set of nice
points is non-empty, which provides the last ingredient in the proof
of Lemma~\ref{lem:five-gon}. 

\note{
\section{Concluding remarks}

In this paper we considered the problem of morphing between two
straight-line drawings of a planar triangulation. We showed that one
can morph between these drawings in $O(n^{2})$ steps, where each step
is a unidirectional morph. However, the grid size of the intermediate
drawings was not bounded. It is then a natural question to ask whether
such problem can be solved while guaranteeing that each intermediate
drawing is on a polynomially sized grid. \FNote{A partial result has
  been obtained in~\cite{BarreraCruz14,barrerahaxelllubiw14} where it
  is shown that for the class of Schnyder drawings we can morph
  between any two of them in $O(n^{2})$ steps while preserving
  planarity and where each intermediate drawing is in a $6n\times 6n$
  grid. However, the question of finding a bound on the grid size for
  morphs between drawings outside this class remains open.}}


%
%


\bibliographystyle{abbrv} \bibliography{unidirection}

\begin{thebibliography}{10}

\bibitem{poly-morph}
S.~Alamdari, P.~Angelini, T.~M. Chan, G.~D. Battista, F.~Frati, A.~Lubiw,
  M.~Patrignani, V.~Roselli, S.~Singla, and B.~T. Wilkinson.
\newblock Morphing planar graph drawings with a polynomial number of steps.
\newblock In {\em Proceedings of the Twenty-Fourth Annual ACM-SIAM Symposium on
  Discrete Algorithms (SODA '13)}. ACM Press, 2013.

\bibitem{Angelini14}
P.~Angelini, G.~{Da Lozzo}, G.~{Di Battista}, F.~Frati, M.~Patrignani, and
  V.~Roselli.
\newblock Morphing planar graph drawings optimally.
\newblock In {\em Proc. forty-first International Colloquium on Automata,
  Languages and Programming (ICALP '14)}, pages 126--137, 2014.

\bibitem{ASS}
B.~Aronov, R.~Seidel, and D.~L. Souvaine.
\newblock On compatible triangulations of simple polygons.
\newblock {\em Computational Geometry: Theory and Applications}, 3:27--35,
  1993.

\bibitem{Barequet04}
G.~Barequet, M.~T. Goodrich, A.~Levi-Steiner, and D.~Steiner.
\newblock Contour interpolation by straight skeletons.
\newblock {\em Graphical Models}, 66(4):245--260, 2004.

\bibitem{Barequet96}
G.~Barequet and M.~Sharir.
\newblock Piecewise-linear interpolation between polygonal slices.
\newblock {\em Computer Vision and Image Understanding}, 63(2):251--272, 1996.

\bibitem{BarreraCruz14}
F.~Barrera-Cruz.
\newblock {\em Morphing planar triangulations}.
\newblock PhD thesis, University of Waterloo, 2014.

\bibitem{barrerahaxelllubiw14}
F.~Barrera-Cruz, P.~Haxell, and A.~Lubiw.
\newblock Morphing schnyder drawings of planar triangulations.
\newblock In {\em Graph Drawing - 22nd International Symposium, {GD} 2014,
  W{\"u}rzburg, Germany, September 24--26, 2014, Revised Selected Papers}, to
  appear.

\bibitem{biedl2013}
T.~Biedl, A.~Lubiw, M.~Petrick, and M.~Spriggs.
\newblock Morphing orthogonal planar graph drawings.
\newblock {\em ACM Transactions on Algorithms (TALG)}, 9(4):29, 2013.

\bibitem{Cairns}
S.~Cairns.
\newblock Deformations of plane rectilinear complexes.
\newblock {\em The American Mathematical Monthly}, 51(5):247--252, 1944.

\bibitem{Floater1999}
M.~S. Floater and C.~Gotsman.
\newblock How to morph tilings injectively.
\newblock {\em Journal of Computational and Applied Mathematics},
  101(1):117--129, 1999.

\bibitem{Gomes99}
J.~Gomes.
\newblock {\em Warping and morphing of graphical objects}, volume~1.
\newblock Morgan Kaufmann, 1999.

\bibitem{Gotsman2001}
C.~Gotsman and V.~Surazhsky.
\newblock Guaranteed intersection-free polygon morphing.
\newblock {\em Computers {\&} Graphics}, 25(1):67--75, 2001.

\bibitem{Kobourov2008}
S.~G. Kobourov and M.~Landis.
\newblock Morphing planar graphs in spherical space.
\newblock {\em J. Graph Algorithms Appl.}, 12(1):113--127, 2008.

\bibitem{LP}
A.~Lubiw and M.~Petrick.
\newblock Morphing planar graph drawings with bent edges.
\newblock {\em J. Graph Algorithms and Applications}, 15(2):205--207, 2011.

\bibitem{T}
C.~Thomassen.
\newblock {Deformations of plane graphs}.
\newblock {\em Journal of Combinatorial Theory}, Series B: 34(3):244--257,
  1983.

\bibitem{Tutte1963}
W.~T. Tutte.
\newblock How to draw a graph.
\newblock {\em Proc. London Math. Soc}, 13(3):743--768, 1963.

\end{thebibliography}

\end{document}